\documentclass[lettersize,journal]{IEEEtran}
\usepackage{amsmath,amsfonts}
\usepackage{algorithmic}
\usepackage{algorithm}
\usepackage{array}
\usepackage[caption=false,font=normalsize,labelfont=sf,textfont=sf]{subfig}
\usepackage{textcomp}
\usepackage{stfloats}
\usepackage{url}
\usepackage{verbatim}
\usepackage{graphicx}
\usepackage{cite}
\usepackage{enumitem} 
\usepackage{amsthm}
\newtheorem{theorem}{Theorem}
\usepackage[colorlinks=true, linkcolor=black, citecolor=blue, urlcolor=blue, hypertexnames=false]{hyperref}

\begin{document}

\title{Joint Antenna Selection and Beamforming Design for Active RIS-aided ISAC Systems}

\author{Wei Ma,~\IEEEmembership{Student Member,~IEEE,}
	 Peichang Zhang,~\IEEEmembership{Member,~IEEE,}
	 Junjie Ye,~\IEEEmembership{Student Member,~IEEE,}
	 \\Rouyang Guan,
	 Xiao-Peng Li,~\IEEEmembership{Member,~IEEE,}
	 Lei Huang,~\IEEEmembership{Senior Member,~IEEE}
	 
\thanks{This work was supported by the Project of Department of Education of Guangdong Province (No. 2021KCXTD008), and in part by the Foundation of Guangdong Key Areas of "Service for Rural Revitalization Plan" (No. 2019KZDZX2014,2020ZDZX1037), and in part by the Foundation of Shenzhen (No.  20200823154213001).

Wei Ma, Peichang Zhang, Junjie Ye, Rouyang Guan, Xiao-Peng Li and Lei Huang are with the State Key Laboratory of Radio Frequency Heterogeneous Integration (Shenzhen University), Shenzhen University, Shenzhen 518060, China (e-mail: 2200432084@
email.szu.edu.cn; pzhang@szu.edu.cn; 2152432003@email.szu.edu.cn; 2200432088@email.szu.edu.cn; x.p.li.frank@gmail.com; lhuang@szu.edu.cn). 

}

}



\maketitle

\begin{abstract}
Active reconfigurable intelligent surface (A-RIS) aided integrated sensing and communications (ISAC) system has been considered as a promising paradigm to improve spectrum efficiency. However, massive energy-hungry radio frequency (RF) chains hinder its large-scale deployment. To address this issue, an A-RIS-aided ISAC system with antenna selection (AS) is proposed in this work, where a target is sensed while multiple communication users are served with specifically selected antennas. Specifically, a cuckoo search-based scheme is first utilized to select the antennas associated with high-gain channels. Subsequently, with the properly selected antennas, the weighted sum-rate (WSR) of the system is optimized under the condition of radar probing power level, power budget for the A-RIS and transmitter. To solve the highly non-convex optimization problem, we develop an efficient algorithm based on weighted minimum mean square error (WMMSE) and fractional programming (FP). Simulation results show that the proposed AS scheme and the algorithm are effective, which reduce the number of RF chains without significant performance degradation. 
\end{abstract}


\section{Introduction}
\IEEEPARstart{I}{n} the next generation mobile communications, spectrum resources are expected to become increasingly scarce due to the explosive growth of communication devices, challenging the improvement of spectrum efficiency \cite{zpc}. To address it,  a potential technology, integrated sensing and communications (ISAC)  has garnered significant attentions and been studied.  

The concept of ISAC suggests that the relatively abundant radar spectrum can be shared with the congested communication sector \cite{terminologies}. It is possible to achieve functions simultaneously within the same spectrum using the same equipment, since the radar and communication are typically similar in terms of hardware and signal processing. The ISAC can be achieved in different levels of integration, including coexistence, cooperation, and co-design \cite{isac_3_theme}. In coexistence, radar and communication systems merely share the same spectrum, primarily focusing on mitigating mutual interference \cite{coexistence_he}. For further interference cancellation, the two systems cooperate and share knowledge to assist in design \cite{cooperation}. The co-design aims to design a unique  dual-functional radar communication (DFRC) systems, which shares the spectrum, the platform and waveform \cite{waveform_design}. Due to its high level of integration, the design of DFRC has garnered extensive attentions. 

Building upon the foundation of ISAC, researches has explored ways to enhance systems performance and address existing challenges. Although ISAC achieves unified sensing and communication by sharing radar spectrum and hardware, its adaptability and resource utilization in complex environments still remain improvable. At this point, the introduction of reconfigurable intelligent surfaces (RIS) offers a novel solution for enhancing ISAC systems. An RIS, composed of numerous controllable reflecting elements, dynamically optimizes the propagation environment by adjusting the phase of incident signals \cite{RISsuveryLYW}. By integrating RIS into ISAC systems, it is possible to improve spectrum utilization, signal quality and system reliability, further advancing ISAC technology.

However, the performance of conventional RIS is often constrained by the path loss. Specifically, the equivalent total path loss of the transmitter-RIS-receiver link is the product of the path losses of the transmitter-RIS and RIS-receiver links, known as the multiplicative fading effect. As a result, RIS performance gains are often limited in several scenarios, restricting its practical applicability. To address this issue,  Active RIS (A-RIS) has been proposed. By incorporating amplifiers into the phase shifters, A-RIS can simultaneously amplify the incident signals and adjust phase shifts \cite{ARISsuvery}. This design effectively reduces the path loss and alleviates the impact of the multiplicative fading effect, providing a more efficient technological solution for ISAC systems.

Finally, for the practical deployment of RIS-assisted ISAC systems, the hardware cost and power consumption of the system require careful consideration. Traditional RIS-aided ISAC systems with a dedicated radio frequency (RF) chain at each antenna incur substantial hardware costs and significant energy consumption. To mitigate these challenges, numerous strategies have been proposed in \cite{AS1}, \cite{AS2}, and \cite{AS3}. Among these, antenna selection (AS) has emerged as a promising method and has garnered significant attention. With the advancements in antenna manufacturing technologies, the cost of antennas has become remarkably low, whereas the cost of RF chains remains high \cite{COST}. By offering a greater number of antenna combination options for the entire system, AS can achieve enhanced performance while maintaining a relatively low overall system cost.

Building upon the aforementioned works, this study aims to achieve enhanced communication performance within the A-RIS-assisted ISAC systems while simultaneously minimizing the number of RF chains to reduce system costs. Specifically, the contributions of this article are as follows:
\begin{enumerate}[label=\arabic*)]
	
	\item \textbf{Introduction of AS to Reduce RF Chains Complexity:} A dedicated AS scheme is proposed to reduce the number of RF chains in the system while maintaining robust communication and sensing performance. Moreover, the scheme achieves superior WSR  with the same number of antennas.
	
	\item \textbf{Joint Optimization of Multiple Components in A-RIS-aided DFRC Systems:} In the A-RIS-aided DFRC system, a joint optimization problem is formulated, which jointly optimizes the transmit AS scheme, transmit beamformer, amplification matrix, and phase shift matrix of the A-RIS. The optimization is performed is subject to constraints such as including radar detection power, constant-modulus transmit waveform, transmit power budget, and A-RIS power budget.
	
	\item \textbf{Stepwise Solution and Alternating Optimization for the Model Problem:}
	The problem is solved using a two-step approach. First, a superior antenna set is selected, and then the transmit beamformer and the A-RIS beamformer are alternately optimized in an iterative manner to efficiently maximize WSR. Specifically, we optimize the DFRC transmit beamformer when the amplification matrix and phase shift matrix of the A-RIS are fixed. We convert the original WSR maximization problem to the weighted minimum meansquare error (WMMSE) problem and then transform the problem into a quadratic constraint quadratic problem (QCQP) problem, which can solved by semi-definite relaxation (SDR).When the transmit beamformer is fixed, the problem becomes logarithmic fractional programming (FP) that can be converted into a standard QCQP problem and easily solved.
	
	\item \textbf{Validation of the Proposed Design's Effectiveness:} Numerical results validate the effectiveness of the proposed AS scheme and beamforming design in A-RIS-assisted ISAC, demonstrating significant performance improvements and providing new insights for the design of similar systems.
\end{enumerate}

The structure of this paper is as follows. Section \ref{Related_work} provides a summary of the related work. Section \ref{model_formulation} introduces the system model. Section \ref{algorithm} describes the proposed algorithm in detail, where the transmit beamformer for DFRC is designed with cuckoo search-based AS, based on the problem formulation and the beamforming matrix for A-RIS. Section \ref{simulation} presents the simulation settings and results. Finally, Section \ref{conclusion} presents the conclusions.

\textbf{Notation}. Boldface lowercase and uppercase letters denote vectors and matrices, respectively. $\left(\cdot \right)^H$, $\left(\cdot \right)^T$ and  $\left(\cdot \right)^*$ represent the conjugate transpose, transpose, and conjugate operators, respectively. The operator $\text{tr}\left(\cdot \right)$ stands for the trace of a matrix. The operator $\text{diag}\left(\mathbf{A} \right)$ attains a vector whose entries are the diagonal elements of a matrix, while $\text{diag}\left(\mathbf{a} \right)$ obtains a diagonal matrix whose diagonal elements are the elements in the vector $\mathbf{a}$. The symbols $\left|\cdot\right|$ and $\left\|\cdot \right\|$ denote absolute value and norm operations, respectively. We use $\mathbb{C}^{M \times N}$ for the complex space of $M \times N$ dimensions. $\mathcal{CN} \left( 0, \sigma^2 \right)$ denotes that a random variable satisfies a  complex Gaussian distribution with zero mean and $\sigma^2$ variance. Finally, we refer to $\mathfrak{R}\left(x \right)$ and $\angle x$ as the real part of the complex number $x$ and the phase of the complex number $x$, respectively.

\section{RELATED WORK}\label{Related_work}
Achieving the coexistence of ISAC poses a key challenge in mitigating mutual interference, for which several solutions have been proposed, such as opportunistic access \cite{Opportunistic_access}, null space projection \cite{NSP}, transceiver design \cite{2018_rihan}. However, coexistence merely involves spectrum sharing, with other resource types remaining separately utilized. To enable deeper integration, DFRC systems have been studied to perform both sensing and communication using the shared equipment \cite{ISAC_survey_big}.  In \cite{Joint_Transmit_Beamforming}, the authors considered transmit beamforming design, where the weighted sum of independent radar waveforms and communication symbols were optimized. Similarly, \cite{Multibeam} proposed the adaptive construction of multiple beams to simultaneously support communication and illuminate targets.  Additionally, a beamforming design for the DFRC systems at  millimeter wave band was proposed to approach the radar beampattern to a desired target and ensure the quality of service (QoS) for communication users \cite{Hybrid_Beamforming_mmwave_isac}.

Additionally, extensive researches has also been conducted to explore the benefits of RIS. In \cite{hexin_ris}, the authors incorporated RIS in the communication systems to minimize the transmit power, while in \cite{sum-rate}, the sum-rate maximization of a RIS-aided multi-user scenario was studied. Moreover, RIS can benefit various scenarios, such as physical layer security \cite{sj3}, near-field communication \cite{Yuanwei_nearfield}, and mobile edge computing \cite{MEC_cunhua}.  However, all of these studies utilized passive RIS (P-RIS) and performance improvement are less noticeable due to the multiplicative fading effect. To tackle this, A-RIS-related works have attracted great attentions. In \cite{active2}, the authors jointly optimized amplification and reflection coefficients of A-RIS to achieve receiving power maximization and RIS-related noise minimization. Besides, the hardware equipment experiments on  A-RIS were conducted in \cite{2021activeris}. In \cite{fully_sub}, a sub-connected architecture of A-RIS was proposed for power saving. In addition, authors demonstrated that A-RIS could improve the spatial diversity in radar detection \cite{active_radar}.

To fully take advantages of ISAC and RIS, numerous endeavors explored the potentials of RIS-aided ISAC systems. In \cite{zpc}, authors jointly designed the DFRC precoder and the RIS beamformer to maximize the radar signal-to-noise ratio (SNR) under a condition of the guaranteed communication user SNR. To reduce the computational burden of the RIS-aided ISAC beamforming design, a heuristic method \cite{Heuristic_RIS_ISAC} and an unsupervised learning approach \cite{UDL_JJ} were developed, respectively. Moreover, a more sophisticated RIS-aided ISAC systems was considered in \cite{LiuRang_RIS_ISAC}, where space-time adaptive processing for the systems design was developed. However, all of these studies focused on P-RIS, leaving multiplicative fading effect remained unresolved. To tackle this, some efforts were made to develop A-RIS-aided ISAC systems. For instance, the work in \cite{ARIS_ISAC} optimized the beamforming of the transmitter and the A-RIS to maximize the radar SNR under different communication metrics. Work \cite{JJ_EE_TVT} further studied an energy efficiency optimization problem in A-RIS-aided ISAC systems. In \cite{ARIS_secure}, employing A-RIS in ISAC systems to improve the physical layer security was also investigated. 

The aforementioned works may still suffer from high power consumption and hardware costs in the base station (BS) due to the deployment of multiple RF chains. Specifically, each data stream requires a dedicated RF chain, and activating all antennas results in significant energy consumption, as each RF chain is energy-hungry. One proper way to solve this problem is AS, as one can activate a subset of antennas with the favorable channel conditions \cite{AS_zhang1}, \cite{AS_zhang2}. Several works have studied what benefits the AS can bring to the ISAC systems \cite{AS_hongbin}, \cite{AS_liurang}, while several works discussed how the AS can aid RIS systems\cite{Xiashang}, \cite{RIS_AS}. Nevertheless, there are few works to discuss the advantages of employing AS to A-RIS aided-ISAC systems.

\section{System Model and Problem Formulation}\label{model_formulation}

\subsection{System Model}

\begin{figure}[!t]
	\centering
	\includegraphics[width=8.89cm, height=6.67cm]{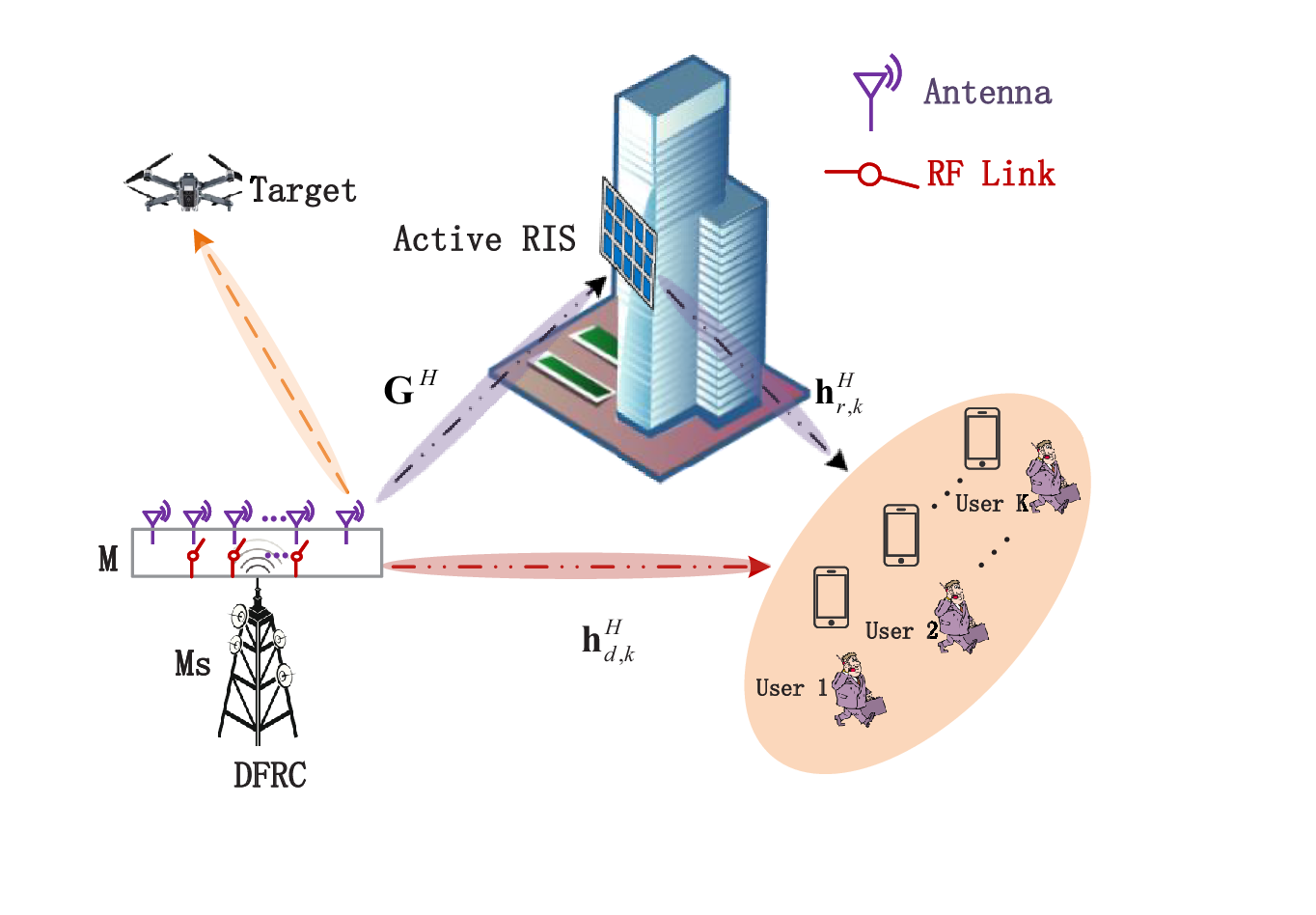}%
	\caption{An A-RIS-aided DFRC system model with the transmitter equipped with $M$ antenna and $M_s$ RF chains.}
	\label{model}
\end{figure}

An A-RIS-aided ISAC system incorporating AS in the BS is illustrated in Fig.\ref{model}. Specifically, $M$ antennas are equipped in the BS, with only $M_s$ RF chains available, indicating that $M_s$ out of $M$ antennas are selected for actual signal transmission. The system is designed to provide services for $K$ single antenna communication users and sense a target. Additionally, an A-RIS with $N$ elements is installed close to the communication users for communication enhancement, where all necessary information can be accessed by the BS and the A-RIS. It is also considered that the target is situated far away from the serving area of A-RIS, such that the A-RIS only plays the role of serving the users. 

In this system, we denote the original communication signal as $\mathbf{x}=[x_1,x_2,\cdots,x_K]^T \in \mathbb{C}^{K}$ with $x_k$ being the symbol for the user $k$. The symbols are mutually independent, i.e. $\mathbb{E}\left[x_k x_k^*\right]=1$ and $\mathbb{E}\left[x_k x_j^*\right]=0$ for $\forall k\neq j$. The signal $\mathbf{x}$ is fed to the selected $M_s$ antennas for propagation through the RF chains. As a result, the output equivalent baseband signal of the RF chains can be $\mathbf{S}=\mathbf{T}\mathbf{x}$, where $\mathbf{T} =[t_1, t_2, \cdots, t_K] \in \mathbb{C}^{M_s \times K}$ is a transmit beamformer. The signal $\mathbf{S}$ then propagates through the channels to serve the communication users and illuminate the target. Without loss of generality, we denote the channels from DFRC to the RIS associated with all $M$ antennas as $\widetilde{\mathbf{G}} \in \mathbb{C}^{N \times M}$, while $\mathbf{G} \in \mathbb{C}^{N \times M_s}$ selects $M_s$ columns from $\widetilde{\mathbf{G}}$ via an AS scheme $\mathcal{P}$ i.e., $\mathbf{G}=\mathcal{P}(\widetilde{\mathbf{G}})$. Similarly, the entire channels from DFRC to the users are formulated as $\widetilde{\mathbf{h}}_{d,k}^H \in \mathbb{C}^{1 \times M }$, while the selected channels are represented by $\mathbf{h}_{d,k}^H \in \mathbb{C}^{1 \times M_s }$, i.e., $\mathbf{h}_{d,k}^H=\mathcal{P}(\widetilde{\mathbf{h}}_{d,k}^H)$. After the signal is emitted to the A-RIS through the selected channels $\mathbf{G}$, it will be tuned in the A-RIS regarding its amplitude and the  phase, after which the tuned signal will propagate through channels $\mathbf{h}_{r,k}^H\in \mathbb{C}^{1 \times N }$ and is received by the users. As a consequence, the signal received by the user $k$ is given by

\begin{equation} \label{com_receive_with_selection}
	y_k = \left( \mathbf{h}_{d,k}^H + \mathbf{h}_{r,k}^H \mathbf{\Theta}^H \mathbf{A}^H \mathbf{G} \right) \mathbf{T} \mathbf{x} 
	+ \mathbf{h}_{r,k}^H \mathbf{\Theta}^H \mathbf{A}^H \mathbf{n}_{1} + n_k,
\end{equation}

\noindent where $\mathbf{A}=\text{diag} \left( \left[ a_1, a_2, \dots, a_N\right] \right)$ is the amplification factors matrix of the A-RIS, while the phase shift matrix of the A-RIS refers to $\mathbf{\Theta}=\text{diag}\left( \left[ e^{j \theta_1}, e^{j \theta_2}, \dots,e^{j \theta_N}\right] \right)$.  The additive white Gaussian noise (AWGN) added at user $k$ is represented by $n_k \sim \mathcal{CN} \left( 0, \sigma_k^2 \right)$, while $\mathbf{n}_{1} \sim \mathcal{CN} \left( \mathbf{0}_N, \sigma_{1}^2\mathbf{I}_N \right)$ is the dynamic noise introduced in the A-RIS.

To evaluate the communication performance of the system, we adopt the WSR of the system. This gives 
\begin{equation}
	R = \sum_{k=1}^{K} \mu_{k} \log_2 \left( 1+\gamma_k  \right),
	\label{WSR}
\end{equation}
in which $\gamma_k$ denotes the signal-to-interference-plus-noise ratio (SINR) of user $k$ and $\mu_k$ is the weight for user $k$. The SINR can be obtained based on the receiving model (\ref{com_receive_with_selection}) and  expressed as
\begin{equation}
	\gamma_k = \frac{ \left|\mathbf{h}_k^H \mathbf{t}_k \right|^2 }{\sum^K_{i=1, i \neq k} \left|\mathbf{h}_k^H \mathbf{t}_i \right|^2 + \sigma_{1}^2 \|\mathbf{h}_{r,k}^H \mathbf{\Theta}^H \mathbf{A}^H\|^2 + \sigma_k^2 }, 
\end{equation}
where $\mathbf{h}_k = \left(\mathbf{h}_{d,k} + \mathbf{G}^H \mathbf{A} \mathbf{\Theta} \mathbf{h}_{r,k}\right) \in \mathbb{C}^{M_s}$ corresponds to the equivalent end-to-end communication channel from DFRC to user $k$. 

In addition to the communications, the DFRC also operates to illuminate a target located at the direction of $\phi$ in a tracking mode. To evaluate this, we resort to the probing power $P_{r}$ in the angle of $\phi$, which can be given by
\begin{equation}
	P_{r} =\frac{\mathbf{a}^H\left(\phi \right) \mathbf{T} \mathbf{T}^H \mathbf{a}\left(\phi \right)}{M_s}.
\end{equation}
The term $\mathbf{T} \mathbf{T}^H$ is the covariance matrix of the transmit beamformer while the vector  $\mathbf{a}\left(\phi \right) \in \mathbb{C}^{M_s}$ is referred to as a selected steering vector towards a direction with angle $\phi$, expressed as
\begin{align}
	\mathbf{a}\left(\phi \right) = \left[e^{\jmath \frac{2\pi d}{\lambda} m_{s,1}}, e^{\jmath \frac{2\pi d}{\lambda} m_{s,2}} , \cdots, e^{\jmath \frac{2\pi d}{\lambda}m_{s,M_s}}\right],
\end{align}
where $m_{s,i} \in \lbrace 0, 1, \cdots, M-1 \rbrace$ indicates the index of the selected antenna linked to $i$-th RF chain ($i \in \lbrace 1, 2, \cdots, M_s \rbrace$). Additionally, one RF chain can only connect to one antenna at the same time, i.e., $m_{s,i}\neq m_{s,j}$ for any $i\neq j$.  The symbol $d$ denotes the spacing of the adjacent antenna, while $\lambda$ is the propagation signal wavelength. Without loss of generalization, the antenna spacing is set as half of $\lambda$. 

\subsection{Problem Formulation}
In this work, we aim to maximize system's WSR while minimizing the number of RF chains. Specifically, the AS scheme $\mathcal{P}$ is designed to select a superior subset of the antennas, while the transmit beamformer $\mathbf{T}$, the amplification matrix $\mathbf{A}$ as well as the reflection matrix $\mathbf{\Theta}$ are jointly designed given the selected antenna. Hence, the optimization problem is formulated as 
\begin{subequations} \label{original_formulation}
	\begin{align}  
		\max _{\mathcal{P},\mathbf{T},\mathbf{\Theta},\mathbf{A}} &~  R \tag{\ref{original_formulation}{a}} \\ 
		s.t.&~ \frac{\mathbf{a}^H\left(\phi,\mathcal{P} \right) \mathbf{T} \mathbf{T}^H \mathbf{a}\left(\phi,\mathcal{P} \right)}{M_s} \geq \eta P_s, \tag{\ref{original_formulation}{b}} \\
		& \text{diag}\left( \mathbf{T} \mathbf{T}^H \right) = \frac{P_s}{M_s} \mathbf{1}_{M_s\times 1},\tag{\ref{original_formulation}{c}}\\
		& \sum_{k=1}^K \left\|\mathbf{\Theta}^H \mathbf{A}^H \mathbf{G}(\mathcal{P}) \mathbf{t}_k\right\|^2 + \left\|\mathbf{A} \mathbf{\Theta} \right\|^2 \sigma_{1}^2 \leq P_a, \tag{\ref{original_formulation}{d}}
	\end{align}
\end{subequations}
where $P_s$ and $P_a$ represent the power budgets of the DFRC and the A-RIS, respectively. The constraint (\ref{original_formulation}{b}) guarantees that the sensing power towards the target direction $\phi$ must be at least $\eta P_s$, where $0 \leq \eta \leq 1$ is the ratio of the transmit power allocated for target sensing in DFRC. The constraint (\ref{original_formulation}{c}) is the constant-modulus constraint imposed by the sensing system to guarantee reduced radar signal distortion caused by the nonlinear amplification process of the transmitter. Meanwhile, the power budget at the DFRC is also guaranteed in (\ref{original_formulation}{c}). 
The constraint (\ref{original_formulation}{d}) indicates that the power consumed by the A-RIS does not exceed the allocated budget $P_a$. The optimization problem (\ref{original_formulation}) is non-convex and thus difficult to solve in its current form. Therefore, the joint optimization necessitates to achieve optimal performance while satisfying all the constraints.

\section{Joint BS and A-RIS beamforming design with AS}\label{algorithm}
In this section, a joint BS and A-RIS beamforming design as well as AS scheme are developed to solve the problem (\ref{original_formulation}). Specifically, we first put forward a cuckoo search-based AS scheme $\mathcal{P}$ to select a superior $M_s$ out of the $M$ antennas associated with the favorable channel conditions. Then, an alternating optimization-based approach is devised to alternately optimize the beamformer of the transmitter and the A-RIS, which utilizes the WMMSE framework and FP. Notably, the amplification matrix $\mathbf{A}$ and the reflecting matrix $\mathbf{\Theta}$ always retain a product form, so we define the matrix $\mathbf{\Psi} = \mathbf{A} \mathbf{\Theta} = \text{diag}\left(a_1 e^{j\theta_1},\dots,a_N e^{j\theta_N} \right) \in \mathbb{C}^{N \times N}$ for convenient expression and design in the following context.

\subsection{Cuckoo Search Policy  for AS}
Considering that the number of the RF chains is less than the antenna number, it is necessary to determine which antennas should be selected before getting to the optimization of other variables. Apparently, searching for a superior antenna subset is a combinatorial optimization problem, which is typically NP-hard \cite{NP_hard}. Exhaustive search can generally yield the optimal result for such problems, but it is computationally expensive, particularly when the total number of the antennas is massive. To reduce the complexity, a discrete cuckoo search-based is proposed for obtaining a superior antenna subset. 

Specifically, to obtain a superior antenna subset associated with the favorable channel conditions, the corresponding cuckoo search-based scheme is presented as follows. First, $L$ antenna subsets ($L$ candidate solutions) are randomly initialized, where $l$-th subset $C_{l}$ contains $M_s$ antennas index given as $C_{l}=\lbrace m_{s,1},  \cdots,m_{s,i}, \cdots, m_{s,M_s}\rbrace$. For each $C_l$, Lévy flight is performed to generate a new candidate solution. The updating rule can be given by
\begin{equation}
	C_{l}^{(t+1)} =C_{l}^{(t)} + \alpha \odot \varepsilon,
	\label{cuckoo_update}
\end{equation}
where $C_{l}^{(t)}$ is the $l$-th solution in $t$-th round search,  $\odot$ represents Hadamard product. To control the search step-size,  $\alpha$ is a step-size scaling factor. The symbol $\varepsilon$ denotes a random step-size following the Lévy distribution, and it can be expressed as
\begin{equation}
	\varepsilon = \frac{u}{|s|^{\frac{1}{\delta}}},
	\label{eq:epsilon}
\end{equation}
where $u$ is a random Gaussian variable following $\mathcal{CN}(0, \sigma_u^2)$. The standard deviation $\sigma_u$  is given by $\sigma_u = \left\{ \frac{\Gamma(1+\delta) \sin\left(\frac{\pi \delta}{2}\right)}{\Gamma\left(\frac{1+\delta}{2}\right) \delta 2^{\frac{\delta-1}{2}}} \right\}^{\frac{1}{\delta}}$, where
$\Gamma$ is the gamma function while $\delta \in (0, 2)$ and $s \sim \mathcal{CN}(0, 1)$. Subsequently, a local random search is performed to discard some of the solutions with a certain probability, while the similarity among different solutions are used to randomly generate a new solution. This is expressed as
\begin{equation}
	C_l^{t+1} = C_l^t + \gamma \odot h(p - \zeta) \odot (C_j^t - C_k^t),
	\label{eq:Tupdate}
\end{equation}
where $\gamma $ and $\zeta$  are random numbers following a uniform distribution, while $p$ is the probability of discarding part of the solutions.  $C_j^t$ and $C_k^t$ are the other solutions and $h(p - \zeta)$ denotes a step function, given as
\begin{align}
	h(p-\zeta) = \begin{cases}
		1, & p - \zeta >0,\\
		0, & p - \zeta \leq 0.
	\end{cases}
\end{align}

To evaluate the obtained solution, a fitness function should be developed. Specifically, we indicate the quality of the antenna subsets by selecting the WSR of (\ref{WSR}) as the fitness function.

By leveraging Lévy flight and the local random search scheme, we update the antenna subsets and evaluate the corresponding channel conditions. It is important to mention that, as the derived solutions may consist of continuous real numbers, floor operation are applied to convert them into integer combinations. Furthermore, the updated antenna subsets might contain duplicate antenna indices, which must be addressed by randomly generating distinct indices to replace duplicates. Additionally, if the obtained antenna index exceeds  $M_s$,  it will be upper-bounded by $M_s$, where similar lower-bounded operation also applies for the case that is less than 1. Moreover, considering that the norm of the subchannels is also affected by different $\boldsymbol{\Psi}$, it should be determined before selecting the antennas. To obtain the $\boldsymbol{\Psi}$ at this stage, please refer to the detailed algorithm refer to \ref{sub_3.3}.
The detailed algorithm procedure is shown in {\textbf{Algorithm 1}}.

\begin{algorithm}[!t] 
	\caption{Proposed Cuckoo Search Algorithm for AS} 
	\begin{algorithmic}[1] 
		\REQUIRE 
		Channels $\widetilde{\mathbf{G}}$, ${\mathbf{h}_{r,k}}$ and $\widetilde{\mathbf{h}}_{d,k}$;
		
		\STATE Randomly initialize $L$, $\boldsymbol{\Psi}$, antenna combinations $C_1, \cdots, C_L$;
		\WHILE{no convergence}
		\STATE Calculate the fitness function value of the corresponding antenna subset $C_l$;
		\STATE Update the antenna subsets with Lévy flight in (\ref{cuckoo_update});
		\IF{the updated subset outperforms the original one}
		\STATE Replace the original with new antenna subset;
		\STATE Check whether the updated solution is legitimate and amend the illegitimate one;
		\ENDIF       
		
		\STATE Update the antenna subsets with local random search in (\ref{eq:Tupdate});
		\IF{the updated subset outperforms the original one}
		\STATE Replace the original with new antenna subset;
		\STATE Check whether the updated solution is legitimate and amend the illegitimate one;
		\ENDIF
		
		\ENDWHILE
		\RETURN The selected antenna index $C^\star=\lbrace m_{s,1},  \cdots,m_{s,i}, \cdots, m_{s,M_s}\rbrace$ and the corresponding subchannels $\mathbf{G}$, and $\mathbf{h}_{d,k}$.
	\end{algorithmic}
\end{algorithm}\label{cukoo_algrithm}

\subsection{WMMSE for Transmit Beamformer Design}
After determining a superior antenna subset, the transmit beamformer $\mathbf{T}$ and the A-RIS beamformer $\mathbf{\Psi}$ are optimized alternating based on the selected antennas. In this subsection, the transmit beamformer $\mathbf{T}$ is optimized with keeping  $\mathbf{\Psi}$ fixed following the WMMSE framework \cite{WMMSE}. 

The WSR maximization problem can be converted into an equivalent WMMSE optimization problem to obtain a tractable form of the problem. To achieve this, a linear estimator $s_k$ is first introduced at user $k$ to estimate the communication symbol from the received signal of (\ref{com_receive_with_selection}), where the estimated symbol can be written as
\begin{align} \label{estimation}
	\hat{x}_k &= s_k y_k  = s_k \left(\mathbf{h}_{k}^H \sum_{i=1}^K \mathbf{t}_i x_i + \mathbf{h}_{r,k}^H \mathbf{\Psi}^H \mathbf{n}_{1} + n_k \right). 
\end{align}
Given the estimated symbol $\hat{x}_k$ in (\ref{estimation}), the  mean square error (MSE) between $\hat{x}_k$  and the true symbol $x_k$ can be derived as
\begin{subequations}\label{sym_err}
	\begin{align} 
		e_k &= \mathbb{E}\left[\|\hat{x}_k-x_k\|^2\right] \\
		&\overset{(a)}{=}|s_k|^2 \bigg(
		\sum_{i=1}^K |\mathbf{h}_{k}^H  \mathbf{t}_i|^2 
		+ \|\mathbf{h}_{r,k}^H \mathbf{\Psi}^H\|^2 \mathbf{\sigma}_{1}^2 
		+\sigma^2_k \bigg) \nonumber \\
		&\quad - 2 \mathfrak{R} \left\lbrace s_k \mathbf{h}_{k}^H  \mathbf{t}_k \right\rbrace 
		+ 1,
	\end{align}  
\end{subequations}

\noindent where the operation (a) leverages the mutual independence of the communication symbols. It can be observed that the MSE of (\ref{sym_err}) is in a quadratic form with respect to $s_k$, thus the optimal estimator $s_k$ for user $k$ is derived by solving $\frac{\partial e_k}{\partial s_k}=0$ to minimize the symbol estimation error. This yields
\begin{align} \label{sMMSE}
	s_k^{MMSE} &= \text{arg}\min_{s_k} e_k \nonumber \\
	&= \frac{\mathbf{t}_k^H \mathbf{h}_k}{
		\sum_{i=1}^K |\mathbf{h}_{k}^H  \mathbf{t}_i|^2 
		+ \|\mathbf{h}_{r,k}^H \mathbf{\Psi}^H\|^2 \mathbf{\sigma}_{1}^2 
		+\sigma^2_k}.
\end{align}

\noindent Accordingly, given the optimal $s_k^{MMSE}$, the corresponding MSE $e_k^{MMSE}$ can then be derived as
\begin{align}
	e_k^{MMSE} &= \mathbb{E}\left[\left\|s_k^{MMSE} y_k - x_k \right\|^2\right] \nonumber \\
	&= 1 - \frac{|\mathbf{h}_k^H \mathbf{t}_k|^2}{
		\sum_{i=1}^K |\mathbf{h}_{k}^H \mathbf{t}_i|^2 
		+ \|\mathbf{h}_{r,k}^H \mathbf{\Psi}^H\|^2 \mathbf{\sigma}_{1}^2 
		+ \sigma^2_k}.
\end{align}
To equivalently transform the original non-convex optimization problem (\ref{original_formulation}) with respect to $\mathbf{T}$ into WMMSE form, we employ the following theorem.

\begin{theorem}
	The WSR problem of transmitter design is equivalent to a weighted MMSE minimization problem when the weighted MMSE coefficients are selected as
	\begin{align} \label{omega}
		\omega_k = \mu_k \left(e_k^{MMSE}\right)^{-1}.
	\end{align}
	With these weighted MMSE coefficients $\omega_k$, the KKT-conditions of the equivalent problem and the original problem can be satisfied simultaneously.
	\label{thm-1-3}
\end{theorem}

\begin{proof}
	Please refer the derivation to \cite{WMMSE}.
\end{proof}

Consequently, the optimization problem (\ref{original_formulation}) can be equivalently transformed into a weighted MSE minimization problem as

\begin{subequations}\label{WMMSE_formulation}
	\begin{align}  
		\min _{\mathbf{T}} &~  \sum_{k=1}^K \omega_k e_k \tag{\ref{WMMSE_formulation}{a}} \\ 
		s.t.&~ \mathbf{a}^H \mathbf{T} \mathbf{T}^H \mathbf{a} \geq \eta M_s P_s, \tag{\ref{WMMSE_formulation}{b}} \\
		& \text{diag}\left(\mathbf{T} \mathbf{T}^H\right) = \frac{P_s}{M_s} \mathbf{1}^{M_s}, \tag{\ref{WMMSE_formulation}{c}}\\
		& \sum_{k=1}^K \left\|\mathbf{\Psi}^H \mathbf{G} \mathbf{t}_k\right\|^2 + \left\|\mathbf{\Psi} \right\|^2 \sigma_{1}^2 \leq P_a. \tag{\ref{WMMSE_formulation}{d}}
	\end{align}
\end{subequations}
Carefully inspecting the problem, it can be found that the objective function (\ref{WMMSE_formulation}{a}) is in a combination of quadratic and linear forms with respective to $\mathbf{t}_k$ while the constraints include quadratic equality constraint as well as quadratic inequality constraint, rendering it non-convex. 

Focusing on the non-convex constraint (\ref{WMMSE_formulation}{b}), its left term can be equivalently rewritten as

\begin{equation}
	\mathbf{a}^H \mathbf{T} \mathbf{T}^H \mathbf{a} = M_s P_s - \sum_{k=1}^K \mathbf{t}_k^H \left(M_s \mathbf{I} - \mathbf{a}\mathbf{a}^H\right) \mathbf{t}_k.
\end{equation}
The constraint (\ref{WMMSE_formulation}{b}) can then be re-expressed as
\begin{equation}
	\sum_{k=1}^K \mathbf{t}_k^H \mathbf{\Bar{Z}} \mathbf{t}_k \leq \left(1-\eta \right) M_s P_s.
	\label{eq15}
\end{equation}
where we denote $\mathbf{\Bar{Z}} = M \mathbf{I} - \mathbf{a}\mathbf{a}^H$. $\mathbf{\Bar{Z}}$ is actually a semi-definite matrix since the $\mathbf{a}\mathbf{a}^H$ is a rank one matrix with the unique non-zero eigenvalue being $M_s$. This makes the constraint (\ref{eq15}) convex. As for the constraint (\ref{WMMSE_formulation}{d}), it can be readily transformed into a standard quadratic constraint by expanding the norm, such that
\begin{align} \label{Y_bar_transform}
	\sum_{k=1}^K \mathbf{t}_k^H \mathbf{\Bar{Y}} \mathbf{t}_k \leq \Tilde{P}_a,
\end{align}
where $\mathbf{\Bar{Y}} = \mathbf{G}^{H} \mathbf{\Psi} \mathbf{\Psi}^{H} \mathbf{G}$, $\Tilde{P}_{a} = P_{a} - \left\|\mathbf{\Psi}\right\|^2 \sigma_{1}^2$. The matrix $\mathbf{\Bar{Y}}$ is obviously semi-definite rendering the constraint (\ref{Y_bar_transform}) convex. To facilitate consistent forms, the constraint (\ref{WMMSE_formulation}{c}) can  be rewritten as
\begin{equation} \label{constraint2}
	\text{diag}\left(\mathbf{T} \mathbf{T}^H\right) = \text{diag}\left(\sum_{k=1}^K \mathbf{t}_k \mathbf{t}_k^H\right) = \frac{P_s}{M_s} \mathbf{1}^{M_s}.
\end{equation}
However, the quadratic equality constraint is still challenging to solve, which requires extra processing later. Observing the objective function (\ref{WMMSE_formulation}{a}), one can notice that it is  actually in a inhomogeneous quadratic form, where the quadratic form and the linear term exists simultaneously. To homogenize and simplify the objective function, we expand and re-express the MSE $e_k$ of (\ref{sym_err}) as follows:
\begin{align}
	e_k &= |s_k|^2 \bigg( 
	\sum_{i=1}^K |\mathbf{h}_{k}^H \mathbf{t}_i|^2 
	+ \|\mathbf{h}_{r,k}^H \mathbf{\Psi}^H\|^2 \mathbf{\sigma}_{1}^2 
	+ \sigma^2_k \bigg) \nonumber \\
	&\quad - 2 \mathfrak{R} \left\lbrace g_k \mathbf{h}_{k}^H \mathbf{t}_k \right\rbrace + 1 \nonumber \\
	&= |s_k|^2 \sum_{i=1}^K \mathbf{t}_i^H \mathbf{h}_{k} \mathbf{h}_{k}^H \mathbf{t}_j 
	- s_k \mathbf{h}_{k}^H \mathbf{t}_k 
	- s_k^* \mathbf{t}_k^H \mathbf{h}_{k} \nonumber \\
	&\quad + |s_k|^2 \sigma^2_k 
	+ |s_k|^2 \|\mathbf{h}_{r,k}^H \mathbf{\Psi}^H\|^2 \mathbf{\sigma}_{1}^2 
	+ 1.
\end{align}

\noindent After dropping the unrelated terms, the objective function can be further simplified as
\begin{align}
	\hat{e}_k &= |s_k|^2 \sum_{i=1}^K \mathbf{t}_i^H \mathbf{h}_{k} \mathbf{h}_{k}^H \mathbf{t}_j 
	- s_k \mathbf{h}_{k}^H \mathbf{t}_k 
	- s_k^* \mathbf{t}_k^H \mathbf{h}_{k} \nonumber \\
	&= \sum_{i=1, i \neq k}^K \mathbf{\Tilde{t}}_i^H \mathbf{C}_{k,1} \mathbf{\Tilde{t}}_i 
	+ \mathbf{\Tilde{t}}_k^H \mathbf{C}_{k,2} \mathbf{\Tilde{t}}_k,
\end{align}

\noindent where $\hat{\mathbf{t}}_k = z_k \mathbf{t}_k$, $\mathbf{\Tilde{t}}_k = \begin{bmatrix}
	\hat{\mathbf{t}}_k & z_k
\end{bmatrix}^T $, $\mathbf{C}_{k,1} = \begin{bmatrix}
	|s_k|^2 \mathbf{h}_k \mathbf{h}_k^H & \mathbf{0} \\
	\mathbf{0}^T & 0
\end{bmatrix}$, $\mathbf{C}_{k,2} = \begin{bmatrix}
	|s_k|^2 \mathbf{h}_k \mathbf{h}_k^H & - s_k^* \mathbf{h}_k \\
	- s_k \mathbf{h}_k^H & 0
\end{bmatrix}$ and $|z_k|^2=1$. Since the dimension of the optimizing variable increases, the constraints (\ref{eq15}), (\ref{Y_bar_transform}), (\ref{constraint2}) should be modified accordingly and given as 
\begin{subequations} \label{modified_constraints}
	\begin{align}
		&\sum_{k=1}^K \mathbf{\Tilde{t}}_k^H \mathbf{\Tilde{Z}} \mathbf{\Tilde{t}}_k \leq \left(1-\eta \right) M P_s,  \\
		& \sum_{k=1}^K \Tilde{\mathbf{t}}_k^H \Tilde{\mathbf{Y}} \Tilde{\mathbf{t}}_k \leq \Tilde{P}_a,\\
		& \text{diag}\left(\sum_{k=1}^K \Tilde{\mathbf{t}}_k \Tilde{\mathbf{t}}_k^H\right) =  \begin{bmatrix}
			\frac{P_s}{M} \mathbf{1}^{M \times 1} \\
			K
		\end{bmatrix}, 
	\end{align} 
\end{subequations}    
where $\mathbf{\Tilde{Z}} =\begin{bmatrix}
	\mathbf{\Bar{Z}} & \mathbf{0} \\
	\mathbf{0}^T & 0
\end{bmatrix}$, 
$\Tilde{\mathbf{Y}}=\begin{bmatrix}
	\mathbf{\Bar{Y}} & \mathbf{0} \\
	\mathbf{0}^T & 0
\end{bmatrix}$ are introduced.

The problem remains non-convex owing to the quadratic equality constraint (\ref{modified_constraints}{c}). To deal with it, we employ the semidefinite relaxation (SDR) approach by  assigning $\mathbf{W}_k$ to be equal to $\mathbf{\Tilde{t}}_k \mathbf{\Tilde{t}}_k^H$ and leveraging the circular property of trace. As a consequence, the problem becomes
\begin{subequations} \label{WMMSE4} 
	\begin{align}
		\min _{\mathbf{W}_1,\dots,\mathbf{W}_K} &~   \sum_{k=1}^K \!\left(\!\omega_k \!\!\!\sum_{i=1,i\neq k}^K \!\!\! \text{tr}\left(\mathbf{C}_{k,1}  \mathbf{W}_i \right) + \omega_k \text{tr}\left(\mathbf{C}_{k,2}  \mathbf{W}_k \right) \right) \tag{\ref{WMMSE4}{a}} \\ 
		s.t.&~ \sum_{k=1}^K \text{tr}\left( \mathbf{\Tilde{Z}} \mathbf{W}_k \right) \leq \left(1-\eta \right) M P_s, \tag{\ref{WMMSE4}{b}}\\
		& \text{diag}\left(\sum_{k=1}^K \mathbf{W}_k \right) =  \begin{bmatrix}
			\frac{P_s}{M} \mathbf{1}^{M} \\
			K
		\end{bmatrix}, \tag{\ref{WMMSE4}{c}}\\
		&  \sum_{k=1}^K\text{tr}\left(\mathbf{\Tilde{Y}} \mathbf{W}_k \right) \leq \Tilde{P}_a, \tag{\ref{WMMSE4}{d}} \\ 
		& \left[\mathbf{W}_k\right]_{M+1,M+1} = 1,  \tag{\ref{WMMSE4}{e}} \\
		& \mathbf{W}_k \succeq 0, \mathbf{W}_k = \mathbf{W}_k^H,  \text{rank}\left(\mathbf{W}_k\right)=1.  \tag{\ref{WMMSE4}{f}}
	\end{align}
\end{subequations}
The rank-one constraint still makes the problem (\ref{WMMSE4}) challenging to solve. One approach to deal with this is to relax the rank-one constraint so that the problem becomes a standard convex problem. However, the obtained $\mathbf{W}_k$ may not satisfy the rank one constraint. In such cases, rank one approximation is required to obtain $\mathbf{t}_k$ by using Gaussian randomization or maximum eigenvalue approximation approach \cite{Rank_one}.

\subsection{Fractional Programming for RIS Beamforming}{\label{sub_3.3}}
This subsection aims to optimize the A-RIS beamforming matrix $\mathbf{\Psi}$ when the transmit beamformer $\mathbf{W}$ is given. With other variables being fixed, the optimization problem (\ref{original_formulation}) can be simplified as 
\begin{subequations}\label{FP}
	\begin{align}  
		\max _{\mathbf{\Psi}} &~  \sum_{k=1}^K \mu_k \log_2 \left( 1+\gamma_k  \right) \tag{\ref{FP}{a}} \\ 
		s.t.&~ \sum_{k=1}^K \left\|\mathbf{\Psi}^H \mathbf{G} \mathbf{t}_k\right\|^2 + \left\|\mathbf{\Psi} \right\|^2 \sigma_{1}^2 \leq P_a. \tag{\ref{FP}{b}}
	\end{align}
\end{subequations}

The problem (\ref{FP}) is a non-convex optimization problem given the fact that the objective function (\ref{FP}{a}) is in a  logarithmic and fractional form with respect to  $\mathbf{\Psi}$. First of all, a Lagrangian dual transform approach \cite{WJZ1}\cite{WJZ2} can be utilized to separate the optimizing variable from the logarithm function. Specifically, we introduce an auxiliary variable $\pmb{\alpha}$ and equivalently transform the problem (\ref{FP}) as
\begin{subequations}\label{FP2}
\begin{align}  
	\max _{\mathbf{\Psi},\pmb{\alpha}} &~  f\left(\mathbf{\Psi},\pmb{\alpha}\right)
	=\sum_{k=1}^K \mu_k \log_2 \left( 1+\alpha_k  \right) \notag \\
	&- \sum_{k=1}^K \mu_k \alpha_k 
	+ \sum_{k=1}^K \frac{\mu_k \left(1+\alpha_k\right) \gamma_k}{1+\gamma_k} \tag{\ref{FP2}{a}} \\ 
	s.t.&~ \sum_{k=1}^K \left\|\mathbf{\Psi}^H \mathbf{G}\mathbf{t}_k\right\|^2 
	+ \left\|\mathbf{\Psi} \right\|^2 \sigma_{1}^2 \leq P_a. \tag{\ref{FP2}{b}}
\end{align}

\end{subequations}
When given $\mathbf{\Psi}$, the problem is unconstrained convex optimization problem with respect to $\pmb{\alpha}$, and thus the optimal $\pmb{\alpha}$ can be obtained by solving the equation $\frac{\partial f\left(\mathbf{\Psi},\pmb{\alpha} \right) }{\partial \alpha_k}=0$. This yields
\begin{equation} \label{update_alpha}
	\alpha_k^\star = \gamma_k.
\end{equation}

Given the optimal $\pmb{\alpha}$, the problem (\ref{FP2}) turns into a FP problem by ignoring the unrelevant terms, given as 
\begin{subequations}\label{FP3}
	\begin{align}  
		\max _{\mathbf{\Psi}} &~  f\left(\mathbf{\Psi}\right) = \sum_{k=1}^K \frac{\mu_k \left(1+\alpha_k\right) \gamma_k}{1+\gamma_k} \tag{\ref{FP3}{a}} \\
		s.t.&~ \sum_{k=1}^K \left\|\mathbf{\Psi}^H \mathbf{G} \mathbf{t}_k\right\|^2 + \left\|\mathbf{\Psi} \right\|^2 \sigma_{1}^2 \leq P_a. \tag{\ref{FP3}{b}}
	\end{align}
\end{subequations}

Apparently, the objective function is in a fraction form, which hinders the problem solving. Fortunately, an equivalent quadratic transformation \cite{FP} can be applied to  the objective function (\ref{FP3}{a}). It has been demonstrated in \cite{FP} that an equivalent form of the original fractional objective function can be given as
\begin{align}\label{FP_transformed_objective_function}
	g \left(\mathbf{\Psi},\pmb{\varepsilon} \right) &= 
	\sum_{k=1}^K 2 \sqrt{\mu_k \left(1+\alpha_k\right)} 
	\mathfrak{R} \left \lbrace \varepsilon_k^{*} \mathbf{h}_k^H \mathbf{t}_k \right \rbrace \nonumber \\
	&\quad - |\varepsilon_k |^2 \left\{ \sum_{i=1}^K \left|\mathbf{h}_k^H \mathbf{t}_i \right|^2 + \sigma_{1}^2 \|\mathbf{h}_{r,k}^H \mathbf{\Psi}^H\|^2 + \sigma_k^2 \right\}.
\end{align}

\noindent where $\pmb{\varepsilon}$ is an introduced auxiliary variable. Similarly, regarding $\pmb{\varepsilon}$, 
the optimal solution $\pmb{\varepsilon}^\star$ can be derived by solving $\frac{\partial g\left(\mathbf{\Psi},\pmb{\varepsilon} \right) }{\partial \varepsilon_k}=0$. As a result, we can obtain
\begin{equation} \label{update_epsilon}
	\varepsilon_k^\star = \frac{\sqrt{\mu_k \left(1+\alpha_k\right)} \mathbf{h}_k^H \mathbf{t}_k}{\sum_{i=1}^K |\mathbf{h}_k^H \mathbf{t}_i|^2 +\sigma_{1}^2 \|\mathbf{h}_{r,k}^H \mathbf{\Psi}^H \|^2 +\sigma_k^2}.
\end{equation}
To further observe the problem, we explicitly express the variable $\mathbf{\Psi}$ by substituting $\mathbf{h}_k^H \!\!= \!\!\left(\!\mathbf{h}_{d,k}^H\!+\!\mathbf{h}_{r,k}^H \mathbf{\Psi}^H \mathbf{G}\!\right)$ into the objective function (\ref{FP_transformed_objective_function}). By expanding the objective function and removing any terms that do not depend on $\mathbf{\Psi}$, one can rewrite the objective function as two linear terms and two quadratic terms:
\begin{align} \label{FP_transformed_objective_function2}
	h \left(\mathbf{\Psi}\right) &= 
	\sum_{k=1}^K 2 \mathfrak{R} \bigg\lbrace 
	\varepsilon_k^{*} \sqrt{\mu_k \left(1+\alpha_k\right)} 
	\mathbf{h}_{r,k}^H \mathbf{\Psi}^H \mathbf{G} \mathbf{t}_k 
	\bigg\rbrace \notag \\
	&\quad - \sum_{k=1}^K 2 \mathfrak{R} \bigg\lbrace 
	|\varepsilon_k |^2 \mathbf{h}_{r,k}^H \mathbf{\Psi}^H \mathbf{G} 
	\sum_{i=1}^K \left( \mathbf{t}_i \mathbf{t}_i^H \right) \mathbf{h}_{d,k} 
	\bigg\rbrace \notag \\
	&\quad - \sum_{k=1}^K |\varepsilon_k |^2 
	\mathbf{h}_{r,k}^H \mathbf{\Psi}^H \mathbf{G} 
	\sum_{i=1}^K \left( \mathbf{t}_i \mathbf{t}_i^H \right) \mathbf{G}^H 
	\mathbf{\Psi} \mathbf{h}_{r,k} \notag \\
	&\quad - \sum_{k=1}^K |\varepsilon_k |^2 \sigma_{1}^2 
	\mathbf{h}_{r,k}^H \mathbf{\Psi}^H \mathbf{\Psi} \mathbf{h}_{r,k}.
\end{align}

\noindent The presence of $\mathbf{\Psi}$ is contained within other matrices as shown in (\ref{FP_transformed_objective_function2}). To make the objective function more manageable, we introduce the variable $\pmb{\psi} = \text{diag}(\mathbf{\Psi})$, and perform matrix transformations based on the following theorem, which allows us to rewrite the objective function in a more tractable form.

\begin{theorem}
	Denote $\mathbf{\Theta} \in \mathbb{C}^{N \times N}$ as a diagonal matrix, $\mathbf{a}, \mathbf{b} \in \mathbb{C}^{N}$ as any vector in an $N \times 1$ dimension space, $\mathbf{D} \in \mathbb{C}^{N \times N}$ as any matrix in an  $N \times N$  dimension space. The following relationships always hold:
	\begin{subequations}
		\begin{align}
			\mathbf{a}^H_{ } \mathbf{\Theta}^H_{ } \mathbf{b} & = \pmb{\theta}^H_{ } \mathbf{A}^H \mathbf{b},\\
			\mathbf{a}^H_{ } \mathbf{\Theta}^H_{ } \mathbf{D}_{ } \mathbf{\Theta}_{ } \mathbf{a} & = \pmb{\theta}^H_{ } \mathbf{A}_{ }^H \mathbf{D}_{ } \mathbf{A}_{ } \pmb{\theta},
		\end{align}
	\end{subequations}
	where $\pmb{\theta} \in \mathbb{C}^{N }$ is a vector whose elements are the elements of $\mathbf{\Theta}$ in diagonal, i.e. $\pmb{\theta}= \text{diag}\left(\mathbf{\Theta}\right)$. The matrix $\mathbf{A}  \in \mathbb{C}^{N \times N}$ is a diagonal matrix whose diagonal elements are the elements in $\mathbf{a}$, i.e. $\mathbf{A}= \text{diag}\left(\mathbf{a}\right)$. 
	\label{thm-2}
\end{theorem}

According to \textbf{Theorem} \ref{thm-2}, the optimization variable $\mathbf{\Psi}$ in the objective function (\ref{FP_transformed_objective_function2}) can be separated from other matrix such that
\begin{align}\label{temp}
	&\sum_{k=1}^K 2 \mathfrak{R} \bigg\lbrace 
	\pmb{\psi}^H \varepsilon_k^{*} \sqrt{\mu_k \left(1+\alpha_k\right)} 
	\text{diag}\left(\mathbf{h}_{r,k}^H\right) \mathbf{G} \mathbf{t}_k 
	\bigg\rbrace \notag \\
	&\quad - \sum_{k=1}^K 2 \mathfrak{R} \bigg\lbrace 
	\pmb{\psi}^H |\varepsilon_k|^2 
	\text{diag}\left(\mathbf{h}_{r,k}^H\right) \mathbf{G} 
	\sum_{i=1}^K \left(\mathbf{t}_i \mathbf{t}_i^H\right) 
	\mathbf{h}_{d,k} 
	\bigg\rbrace \notag \\
	&\quad - \sum_{k=1}^K 
	\pmb{\psi}^H |\varepsilon_k|^2 
	\text{diag}\left(\mathbf{h}_{r,k}^H\right) \mathbf{G} 
	\sum_{i=1}^K \left(\mathbf{t}_i \mathbf{t}_i^H\right) 
	\mathbf{G}^H \text{diag}\left(\mathbf{h}_{r,k}\right) \pmb{\psi} \notag \\
	&\quad - \sum_{k=1}^K 
	\pmb{\psi}^H |\varepsilon_k|^2 \sigma_{1}^2 
	\text{diag}\left(\mathbf{h}_{r,k}^H\right) 
	\text{diag}\left(\mathbf{h}_{r,k}\right) \pmb{\psi}.
\end{align}

To further simplify the expression (\ref{temp}), we combine terms with the same order and introduce symbols $\mathbf{v}$ and $\mathbf{U}$ to represent the coefficients of the linear and quadratic terms, respectively, which yields
\begin{subequations}
\begin{align}
	\mathbf{v} &= \sum_{k=1}^K \bigg(
	\varepsilon_k^{*} \sqrt{\mu_k \left(1+\alpha_k\right)} 
	\text{diag}\left(\mathbf{h}_{r,k}^H\right) \mathbf{G} \mathbf{t}_k \notag \\
	&\quad - |\varepsilon_k|^2 
	\text{diag}\left(\mathbf{h}_{r,k}^H\right) \mathbf{G} 
	\sum_{j=1}^K \left(\mathbf{t}_i \mathbf{t}_i^H\right) \mathbf{h}_{d,k} 
	\bigg), \\
	\mathbf{U} &= \sum_{k=1}^K \bigg(
	|\varepsilon_k|^2 
	\text{diag}\left(\mathbf{h}_{r,k}^H\right) \mathbf{G} 
	\sum_{i=1}^K \left(\mathbf{t}_i \mathbf{t}_i^H\right) \mathbf{G}^H 
	\text{diag}\left(\mathbf{h}_{r,k}\right) \notag \\
	&\quad - |\varepsilon_k|^2 \sigma_{1}^2 
	\text{diag}\left(\mathbf{h}_{r,k}^H\right) 
	\text{diag}\left(\mathbf{h}_{r,k}\right)
	\bigg).
\end{align}

\end{subequations}
\noindent Therefore, the objective function (\ref{FP_transformed_objective_function}) can then be expressed in a more concise form as
\begin{equation}
	2 \mathfrak{R} \left \lbrace \pmb{\psi}^H \mathbf{v} \right\rbrace - \pmb{\psi}^H \mathbf{U} \pmb{\psi}.
	\label{eq36}
\end{equation}
On the other hand, it is observed that the constraint of (\ref{FP}{b}) is also a quadratic term, so it can be handled using \textbf{Theorem} \ref{thm-2}. As a result, it can be rewritten in a standard quadratic constraint form as
\begin{equation}
	\pmb{\psi}^H \mathbf{\Pi} \pmb{\psi} \leq P_a,
	\label{eq38}
\end{equation}
where $\mathbf{\Pi}$ is a semi-definite matrix given by
\begin{equation}
	\mathbf{\Pi} = \sum_{k=1}^K \text{diag}\left(\mathbf{G} \mathbf{t}_k \right) \text{diag}\left(\mathbf{G} \mathbf{t}_k \right)^H + \sigma_{1}^2 \mathbf{I}.
\end{equation}

With (\ref{eq36}) and (\ref{eq38}), the optimization problem can be eventually given as
\begin{subequations}\label{phi_obj2}
	\begin{align} 
		\max _{\pmb{\psi}} &~   2 \mathfrak{R} \left \lbrace \pmb{\psi}^H \mathbf{v} \right\rbrace - \pmb{\psi}^H \mathbf{U} \pmb{\psi}  \\ 
		& s.t.~ \pmb{\psi}^H \mathbf{\Pi} \pmb{\psi} \leq P_a. \tag{\ref{phi_obj2}{b}}
	\end{align}
\end{subequations}
The problem (\ref{phi_obj2}) has the standard form of a convex QCQP problem, which can be solved directly using the CVX solver. After obtaining the optimal $\pmb{\psi}$, we can compute the optimal $\mathbf{\Psi}$ as $\mathbf{\Psi} = \text{diag}(\pmb{\psi})$. The amplification coefficients matrix $\mathbf{A}$ and the phase shift matrix $\mathbf{\Theta}$ can then be computed as $\mathbf{A} = |\mathbf{\Psi}|$ and $\mathbf{\Theta} = \angle \mathbf{\Psi}$, respectively. We summarize the overall algorithm in \textbf{Algorithm 2}.

\begin{algorithm}[!t] 
	\caption{Proposed AS and Beamforming algorithm} 
	\begin{algorithmic}[1] 
		\REQUIRE ~~ 
		Channels ${\mathbf{G}}$, ${\mathbf{h}}_{r,k}$, and ${\mathbf{h}}_{d,k}$, power budget ${P_s}$ and ${P_a}$, target angle ${\phi}$, radar power ratio ${\eta}$, , communication priorities $\mu_{k}$,noise covariance $\sigma_k^2$ and $\sigma_{1}^2$;
		\ENSURE ~~ 
		Optimized $\mathbf{T}$, $\mathbf{A}$, $\mathbf{\Theta}$, sum-rate $R$, radar probing power $P_r$, ${C}$;
		\STATE Initialize $\mathbf{T}$, $\mathbf{A}$, $\mathbf{\Theta}$, $\pmb{\varepsilon}$;
		\STATE Select an antenna subset with \textbf{Algorithm 1};
		\WHILE {no convergence of $R_{\text{sum}}$}
		\STATE Calculate ${\mathbf{s}^{MMSE}}$ by (\ref{sMMSE});
		\STATE Calculate ${\pmb{\omega}}$ by (\ref{omega});
		\STATE Obtain ${\mathbf{W}}$ by solving (\ref{WMMSE4}) and approximate ${\mathbf{T}}$ with rank-one decomposition;
		\STATE Update ${\pmb{\alpha}}$ by solving (\ref{update_alpha});
		\STATE Update ${\pmb{\varepsilon}}$ by solving (\ref{update_epsilon});
		\STATE Update ${\mathbf{\Psi}}$ by solving (\ref{phi_obj2});
		\ENDWHILE	
		\STATE Obtain $\mathbf{A}$ and $\mathbf{\Theta}$ from ${\mathbf \Psi}$;
		\RETURN Optimized $\mathbf{T}$, $\mathbf{A}$, $\mathbf{\Theta}$, $R$, ${P_a}$, $C$. 
	\end{algorithmic}
\end{algorithm}

\subsection{Convergence and Complexity Discussions} \label{convergence complexity}
In this subsection, we analyze the convergence and complexity of the proposed algorithm, including \textbf{Algorithm 1} and \textbf{Algorithm 2}. For this convergence, the convergence of this algorithm is not affected by the Cuckoo Search-based AS \cite{cuckoo}, and the Algorithm 2 uses an iterative alternating optimization routine, where the transmit and reflective beamforming variables are optimized alternatively  in each iteration, i.e. 
\begin{align*}
	\dots \mathbf{T}^k \longrightarrow \mathbf{\Psi}^k \longrightarrow \mathbf{T}^{k+1} \longrightarrow \mathbf{\Psi}^{k+1} \dots .
\end{align*}
Since each step of the iteration, namely (\ref{WMMSE4}) and (\ref{phi_obj2}), are convex problems the convergence of the alternating optimization algorithm is assured according to \cite{convex_boyd}. 

The computational complexity of the proposed algorithm is primarily determined by the combined contributions of Algorithm 1 and Algorithm 2.

The complexity of the Cuckoo Search-based AS algorithm primarily depends on the WSR fitness function. According to (\ref{WSR}), the complexity of the fitness function WSR is $\mathcal{O}\left(K(K M_s(1+N^2+N)+2N^2+M_s) \right)\simeq\mathcal{O}\left(M_sK^2N^2\right)$
During each iteration, the solution set is updated using both the Lévy flight (\ref{cuckoo_update}) and local random search (\ref{eq:Tupdate}). As a result, each iteration requires two computations of the WSR. The total complexity is therefore $\mathcal{O}\left(I_aLM_sK^2N^2\right),$ where $L$ represents the size of the initialized population, and $I_a$ denotes the number of iterations of the algorithm.

The complexity of Algorithm 2 is determined by the updates of $\mathbf{T}$ and $\mathbf{\Psi}$. When designing $\mathbf{T}$, the complexity primarily involves solving the SDR problem (\ref{WMMSE4}) and performing the eigenvalue decomposition, which are $\mathcal{O}\left(({M_s}+1)^{4.5} \log(1/\epsilon)\right)\simeq\mathcal{O}\left(M_s^{4.5} \log(1/\epsilon)\right)$ and $\mathcal{O}\left(K(M_s+1)^3\right)\simeq\mathcal{O}\left(K {M_s}^3\right)$ respectively, wherein $\epsilon$ denotes the desired solution accuracy. Similarly, when obtaining $\mathbf{\Psi}$, the complexity mainly results from the QCQP problem, which incurs the complexities of $\mathcal{O}\left(N^{3} \log(1/\epsilon)\right)$. Overall, the computational complexity for the Algorithm 2 is $\mathcal{O}\left(I_b(({M_s}^{4.5} \log(1/\epsilon)+K {M_s}^3+N^3\log(1/\epsilon))\right),$
where $I_b$ represents the number of iterations in Algorithm 2.
Therefore, the overall complexity of the proposed system algorithm is the sum of the complexities of Algorithm 1 and Algorithm 2, which is 
$\mathcal{O}\left(I_aLM_sK^2N^2+I_b({M_s}^{4.5} \log(1/\epsilon)+K {M_s}^3+N^3\log(1/\epsilon))\right).$

\section{Simulation Results}\label{simulation}
In this section, numerical results are presented to demonstrate the effectiveness of the proposed algorithm for the A-RIS-assisted DFRC system with the cuckoo search-based AS scheme. In the system illustrated in Fig. \ref{simulation_mod}, the DFRC is located at ($0$m, $0$m), while the A-RIS is positioned at ($150$m, $0$m). The DFRC transmits communication symbols to $K\!=\!4$ downlink users located in a region centered at ($150$m, $10$m) with a radius of $5$m. The total system power budget is $P\!=\!20$ dBm, and the noise power levels are defined as $\sigma^2\!=\!-20$ dBm and $\sigma_{1}^2\!=\!-40$ dBm. In addition, the channels related to the RIS, including $\mathbf{G}$ and $\mathbf{h}_{r,k}$, are modeled as Rician channels following the approach in \cite{JJ_EE_TVT}, while the channels between the users and the DFRC, $\mathbf{h}_{d,k}$,  are assumed as Rayleigh fading channels. 

\begin{figure}[h]
	\centering
	\includegraphics[width=8.89cm, height=6.67cm]{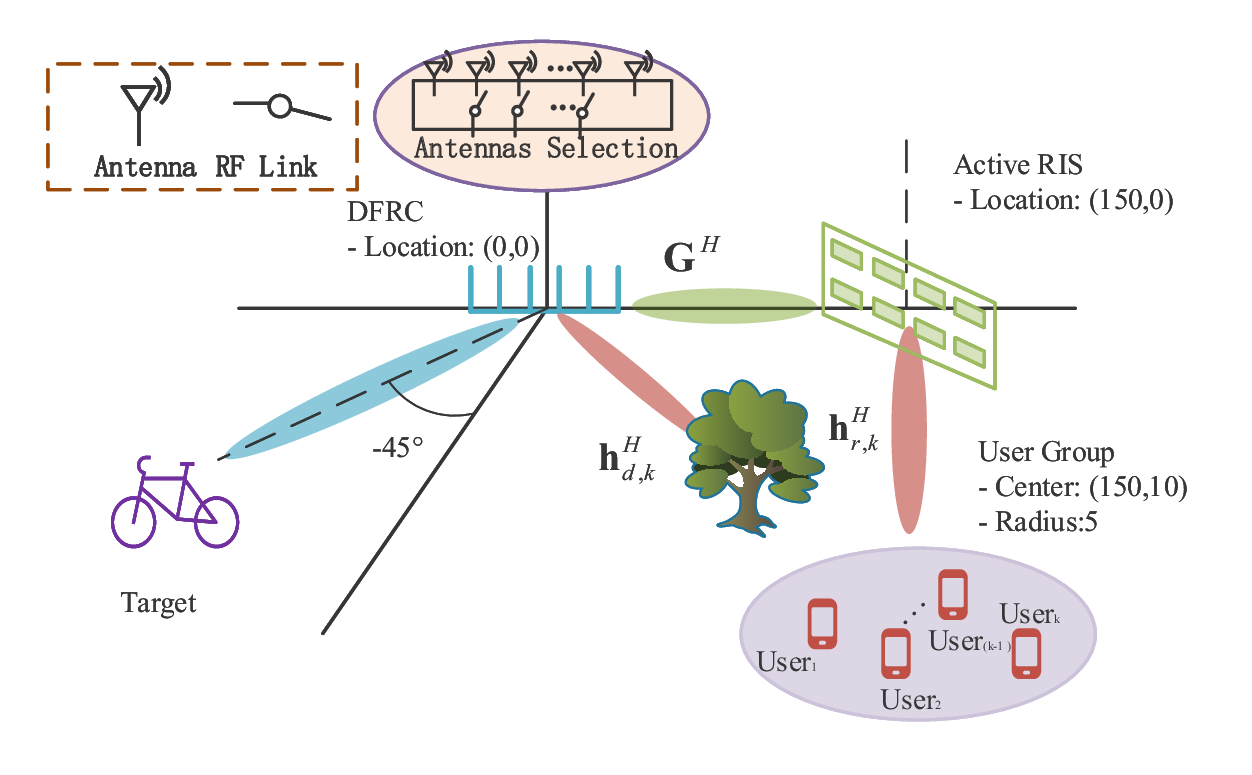}
	\caption{The simulated A-RIS aided DFRC scenario comprising of an DFRC BS with $M$-antenna and $M_s$  RF chains, an $N$-element active RIS, $4$ single antenna users, and one target.}
	\label{simulation_mod}
\end{figure}

To clarify the power allocation in the system and facilitate the demonstrations in the following sections, we present the relationship between symbols related to power and power allocation. The total system power is denoted by $P$, which can be divided into the power budget at the DFRC, denoted by $P_s$, and the power budget at the A-RIS, denoted by $P_a$. The power split ratio between the A-RIS and the DFRC is denoted as $\rho$, such that $P_s=\rho P$ and $P_a=(1-\rho) P$. Moreover, the power of the DFRC $P_s$ can be further divided into the power for radar, denoted by $P_r$, and the power for communication, denoted by $P_c$. The power split ratio for radar in the DFRC is denoted by $\eta$, such that $P_r = \eta P_s$. Therefore, the total power assigned for radar detection is given by $P_r = \eta \rho P$. Unless otherwise specifications, the radar detection power ratio $\eta$ and the power split ratio of the total system power for DFRC/BS are fixed at $0.75$ and $0.9$, respectively. Furthermore, to demonstrate the effectiveness and superiority of the proposed algorithm, three approaches are selected for comparison. The first approach involves the transmitter being equipped with an array whose number of antennas matches that of the AS case. The second approach considers a random AS scenario. The final approach optimizes the original array without any antenna selection.

\begin{figure}[t]
	\centering
	\renewcommand{\figurename}{Fig.} 
	\includegraphics[width=8.89cm, height=6.67cm]{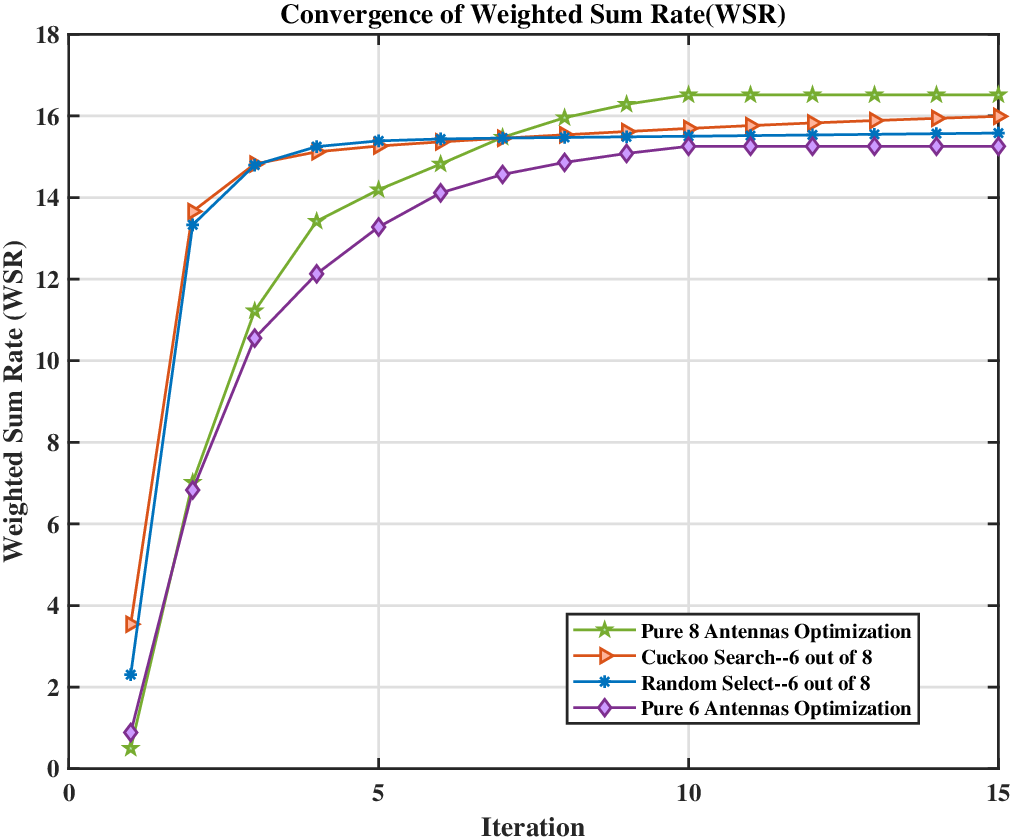} 
	\caption{Convergence behaviour of the proposed algorithm: WSR versus the iteration number.}
	\label{convergence} 
	\renewcommand{\figurename}{Figure} 
\end{figure}

In Fig. \ref{convergence}, the convergence behavior of the proposed joint AS and beamforming algorithm is presented, along with its comparison to three benchmark cases. The total system power $P$ is assumed to be $20$ dBm for all cases. The total number of antennas is $8$, with $6$ antennas selected.  In the scenario, the power split ratio between the A-RIS and DFRC is set to $0.9$. It can be recognized from Fig.\ref{convergence} that all the scenarios can converge. Compared to the case where the system is directly equips with $6$ antenna and the random search case, the WSR utilizing the cuckoo AS has a better performance. This accounts for the fact that the selection provides greater degrees of freedom (DoF) for optimization while the heuristic cuckoo search-based enables finding a superior antenna combination.

\begin{figure}[t] 
	\centering
	\renewcommand{\figurename}{Fig.}
	\includegraphics[width=8.89cm, height=6.67cm]{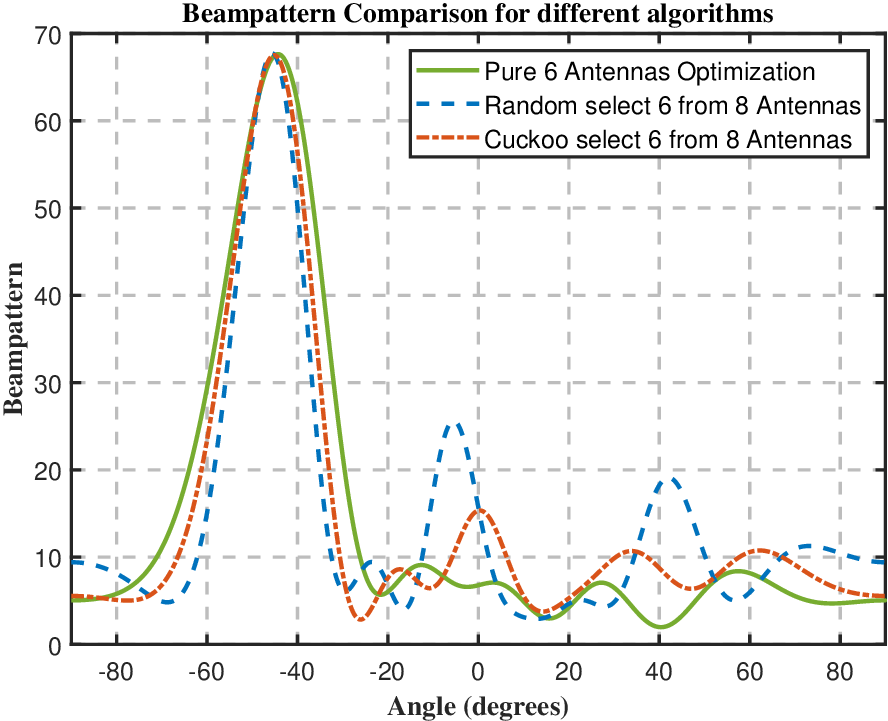}
	\caption{DFRC transmit beampatterns for different algorithms.}
	\label{Bearmpattern}  
	\renewcommand{\figurename}{Figure}
\end{figure}

In Fig.~\ref{Bearmpattern}, the beampatterns of three different algorithms in A-RIS-assisted DFRC systems are presented. In all three scenarios, the total system power  $P$ is fixed at 20 dBm, and the number of reflective elements is 36. The total power budget allocated for radar detection is fixed at 0.75. The power split ratio $\rho$ for the DFRC system is set at 0.9. Both antenna selection methods involve selecting 6 antennas out of a total of 8, while the pure optimization scheme also operates with 6 antennas. It can be observed that in all three cases, the power allocated for the radar function meets the required constraints. Furthermore, the peaks in the direction of the A-RIS verify the effectiveness of the communication. However, it is noteworthy that following antenna selection, the beam sidelobes for both methods are increased. This phenomenon attributes to the disruption of uniformity within the ULA antenna array caused by antenna selection, introducing asymmetry and consequently elevates the sidelobe levels.

\begin{figure}[t]
	\centering
	\renewcommand{\figurename}{Fig.} 
	\includegraphics[width=9.3cm, height=6.7cm]{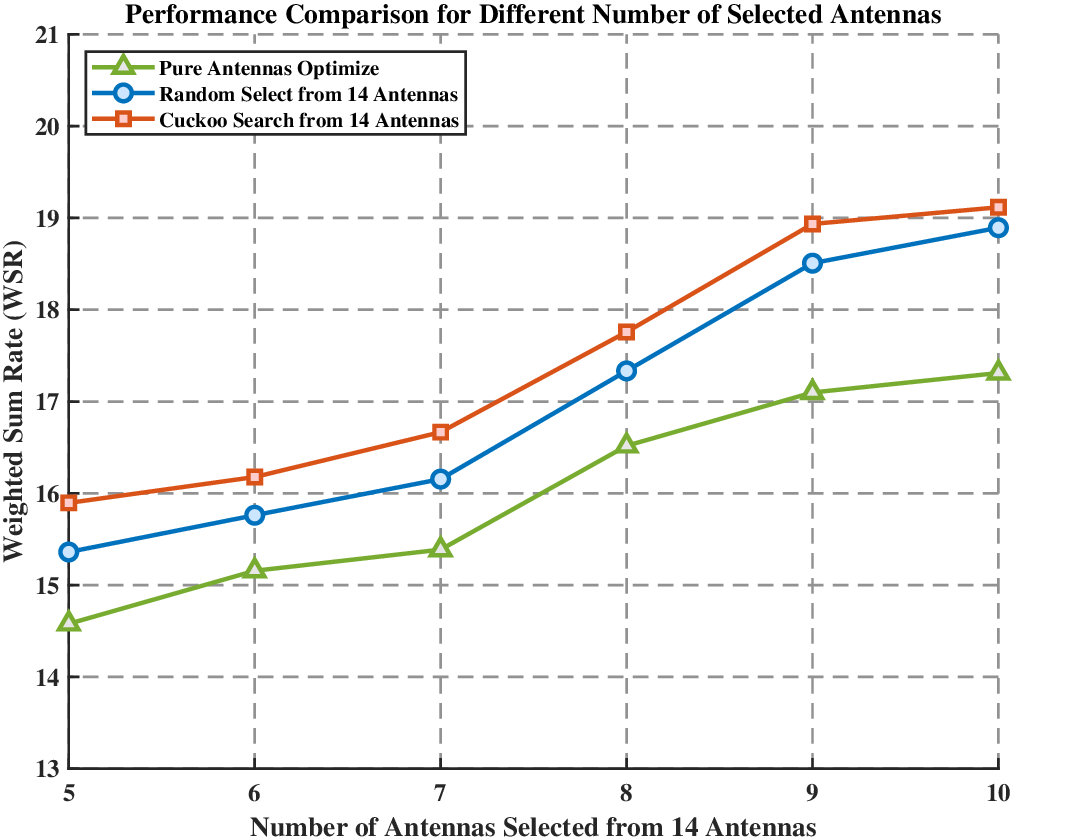} 
	\caption{WSR versus different AS numbers.}
	\label{Antenna_select}
	\renewcommand{\figurename}{Figure} 
\end{figure}

\begin{figure}[t]
	\centering
	\renewcommand{\figurename}{Fig.}
	\includegraphics[width=8.89cm, height=6.67cm]{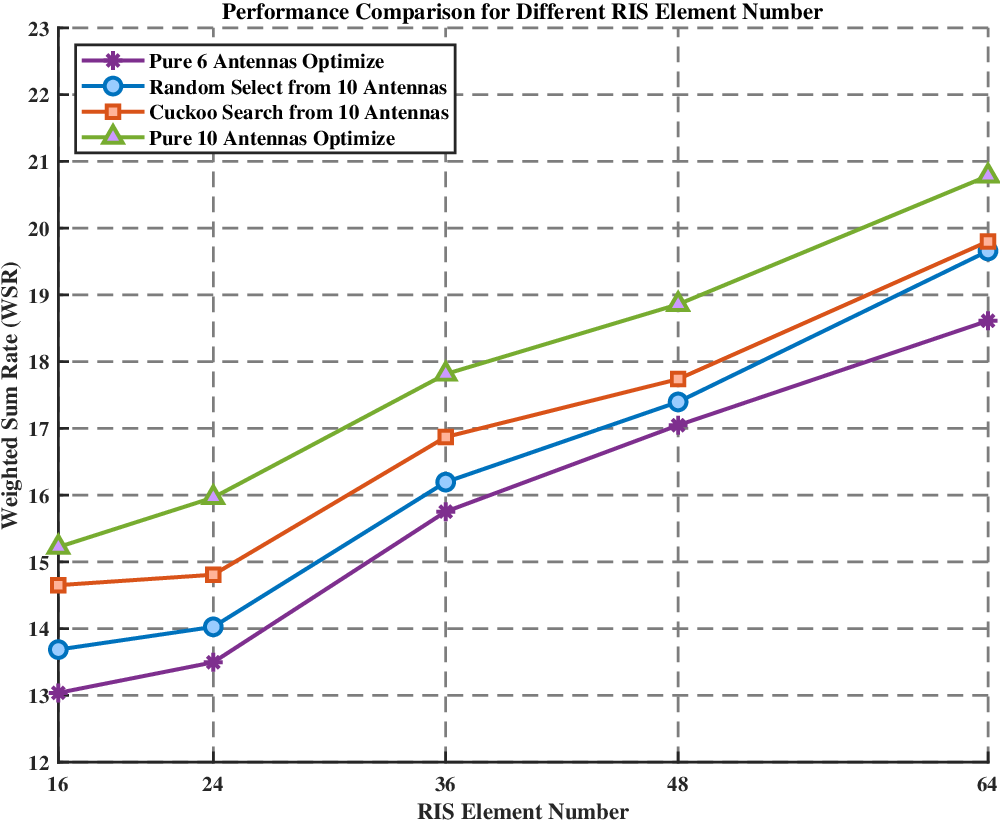}
	\caption{WSR versus the number of reflecting elements $N$.}
	\label{N_effect}  
	\renewcommand{\figurename}{Figure}
\end{figure}

\begin{figure}[t]
	\centering
	\renewcommand{\figurename}{Fig.}
	\includegraphics[width=8.89cm, height=6.67cm]{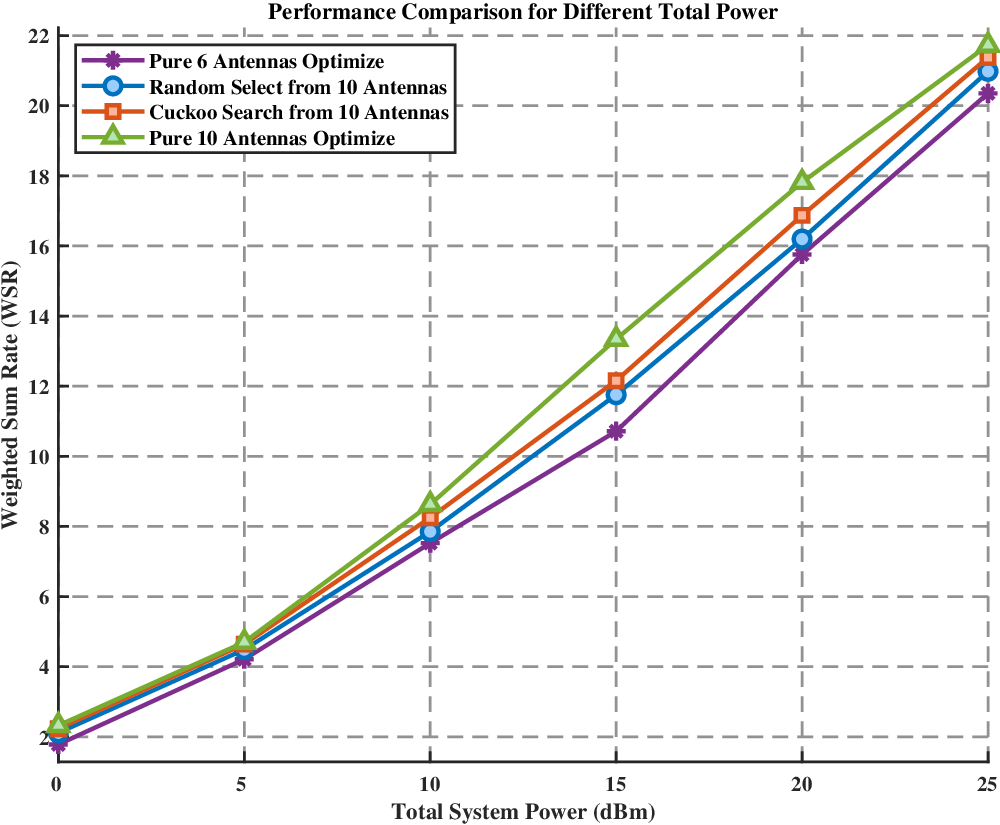}
	\caption{WSR versus total system power $P$.}
	\label{Pt_effect}
	\renewcommand{\figurename}{Figure}
\end{figure}

Fig. \ref{Antenna_select} shows the impact of different number of AS on the WSR. In this experiment, we set the total transmitter antenna number to $14$ and vary the number of selected antennas. As shown in Fig. \ref{Antenna_select}, the WSR increases as the number of selected antennas increases. The reason behind this is that more antennas provide more space diversity yielding a higher WSR. Moreover, under the same number of selected antennas, our proposed cuckoo search-based algorithm achieves the best results followed by random AS, while the case that directly equips the antennas with equal number gives the worst performance. This result is consistent with the previous findings.

Fig. \ref{N_effect} presents the effect of the number of RIS elements on WSR, where the number of reflecting element number varies from $16$ to $64$.  As shown in Fig. \ref{N_effect}, the WSR of all four cases increase as the number of RIS elements increases. This is reasonable since the increase of RIS element enables more DoF to alter the channel conditions. Additionally, consistent with the previous results, the cases involving AS exhibit a higher WSR compared to directly deploying an array with the same number of antennas.

\begin{figure}[t]
	\centering 
	\renewcommand{\figurename}{Fig.} 
	\includegraphics[width=8.89cm, height=6.67cm]{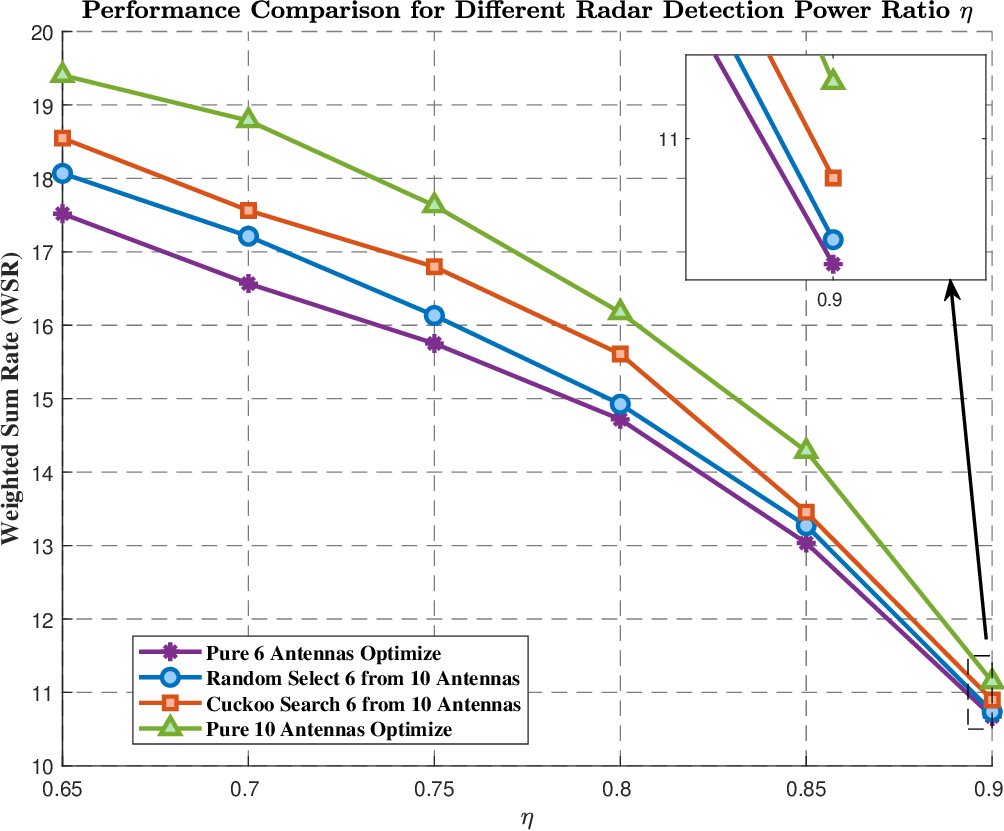}
	\caption{Trade-off between WSR and radar detection power ratio $\eta$.}
	\label{Tradeoff2}
	\renewcommand{\figurename}{Figure} 
\end{figure}

\begin{figure}[t]
	\centering 
	\renewcommand{\figurename}{Fig.} 
	\includegraphics[width=8.89cm, height=6.67cm]{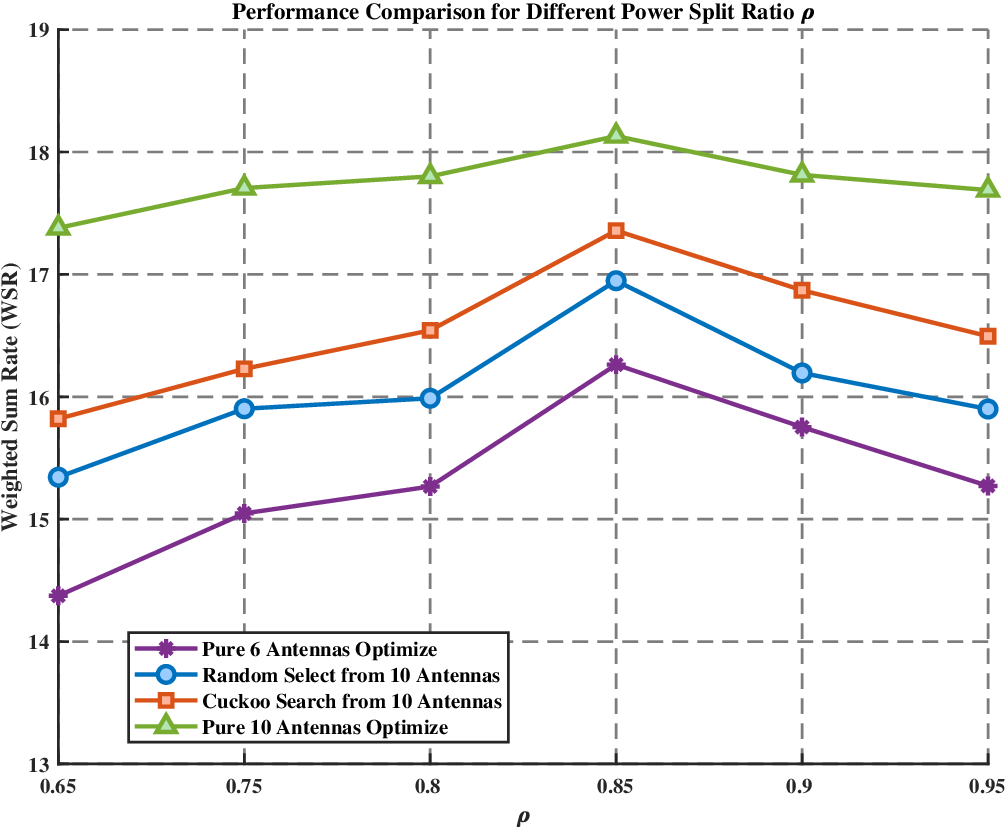}
	\caption{Trade-off between WSR and power split ratio $\rho$.}
	\label{power_split}
	\renewcommand{\figurename}{Figure} 
\end{figure}

Fig. \ref{Pt_effect} illustrates the variation of WSR with increasing system power. In this experiment, the number of RIS elements is fixed at $36$. The results demonstrate that the WSR of all four setups increases as the total system power increases, where the cuckoo algorithm still outperforms the benchmarks. This can be attributed to the fact that, under the same channel conditions, the receiver can achieve higher signal strength and SINR when provided with more transmit power. And the channel corresponding to the antenna chosen by the cuckoo is always superior.

Fig. \ref{Tradeoff2} illustrates the trade-off between required radar detection power and WSR. The total system power is set as $20$ dBm and the number of RIS elements is $36$. The power split ratio $\rho$ for DFRC in the A-RIS case is set to $0.9$. The radar detection power ratio $\eta$ of DFRC ranges from $0.1$ to $0.9$, and the radar detection power is calculated using $P_r = \eta \rho P$. It can be observed that when the required radar detection power increases, the WSR decreases proportionally. This suggests that increasing the power allocated to radar detection in DFRC comes at the cost of reducing the power available for communication, which leads to a lower WSR.

Fig. \ref{power_split} shows the relationship between the power split ratio parameter $\rho$ and WSR. The total system power is fixed at $20$ dBm. Since the total system power $P$ and radar detection power $P_r$ are fixed, $\rho$ can only range from $0.6$ to $1$. It can be observed that WSR initially increases with $\rho$ and then gradually decreases. The decrease in WSR can be attributed to the decrease in power allocated to the A-RIS as $\rho$ increases, which reduces the ability to mitigate the multiplicative fading effect. On the other hand, when $\rho$ is small, the power transmitted from DFRC is limited, leading to a small WSR. However, as the power transmitted for communication in DFRC increases with a higher $\rho$, the WSR correspondingly increases .

\section{Conclusions}\label{conclusion}
In this paper, we present a study on the A-RIS-aided DFRC system hat combines an AS scheme to enhance communication WSR while maintaining radar sensing performance. A joint AS and beamforming optimization algorithm is proposed by leveraging the cuckoo search-based scheme, WMMSE framework as well as the FP. Simulation results validate the effectiveness of the proposed algorithm and indicate that the A-RIS-aided DFRC system with cuckoo AS can reduce the RF chains without significant performance degradation.





\bibliographystyle{ieeetr}
\bibliography{literature.bib}

\end{document}